\documentclass[journal]{IEEEtran}

\usepackage{amsmath}         
\usepackage{amsthm}
\usepackage{amssymb}         
\usepackage{graphicx}        
\usepackage{mathptmx}        
\usepackage{times}           
\usepackage{comment}					
\newtheorem{lem}{Lemma}
\newtheorem{thm}{Theorem}

\newtheorem{corollary}{Corollary}

\newtheorem{definition}{Definition}
\newtheorem{exmp}{Example}

\def\<{\leqslant}           
\def\>{\geqslant}           

\def\wh{\widehat}
\def\wt{\widetilde}
\def\~{\wt{~}}
\def\Re{\mathrm{Re}}   
\def\Im{\mathrm{Im}}   

\def\cH{\mathcal{H}}     
\def\mA{\mathbb{A}}      
\def\mR{\mathbb{R}}      
\def\mC{\mathbb{C}}      
\def\fH{\mathfrak{H}}    
\def\rT{\mathrm{T}}        
\def\bS{\mathbf{S}}
\def\bE{\mathbf{E}}    
\def\[[[{[\![\![}
\def\]]]{]\!]\!]}

\def\rd{{\rm d}}        

\def\bD{\mathbf{D}}

\def\bN{\mathbf{N}}
\def\bR{\mathbf{R}}

\def\bJ{\mathbf{J}}

\def\x{\times}
\def\ox{\otimes}

\def\fA{\mathfrak{A}}

\def\fF{\mathfrak{F}}

\def\fa{\mathfrak{a}}

\def\bH{{\mathbf H}}

\def\bX{{\bf X}}
\def\bM{{\bf M}}

\def\cF{\mathcal{F}}

\def\bJ{\mathbf{J}}

\def\mH{\mathbb{H}}
\def\mS{\mathbb{S}}

\def\diag{\mathop{\mathrm{diag}}}    

\DeclareMathAlphabet{\mathbfit}{OML}{cmm}{b}{it}
\DeclareMathAlphabet      {\mathbfit}{OML}{cmm}{b}{it}
\DeclareMathAlphabet      {\mathbfd}{OT1}{cmr}{bx}{n}

\begin{document}

\title{A Modified Frequency Domain Condition for \\the Physical Realizability of \\ Linear Quantum Stochastic Systems}

\author{Arash~Kh.~Sichani,~\IEEEmembership{Member,~IEEE,}		
        Ian~R.~Petersen,~\IEEEmembership{Fellow,~IEEE}       
        
\thanks{This work is supported by the Australian Research Council. The authors are with UNSW Canberra, ACT 2600, Australia. E-mail: {\tt arash\_kho@hotmail.com, i.r.petersen@gmail.com}.}
}

\maketitle

\begin{abstract}
This note is concerned with a modified version of the frequency domain physical realizability (PR) condition for linear quantum systems. We consider open quantum systems whose dynamic variables satisfy the canonical commutation relations of an open quantum harmonic oscillator and are governed by linear quantum stochastic differential equations (QSDEs). In order to correspond to physical quantum systems, these QSDEs must satisfy PR conditions. We provide a relatively simple proof that the PR condition is equivalent to the frequency domain $(J,J)$-unitarity of the input-output transfer function and orthogonality of the feedthrough matrix of the system without the technical spectral assumptions required in previous work. We also show that the poles and transmission zeros associated with the transfer function of PR linear quantum systems are the mirror reflections of each other about the imaginary axis. An example is provided to illustrate the results.
\end{abstract}


\section{Introduction}
Quantum stochastic differential equations (QSDEs) \cite{HP_1984,P_1992} provide a framework for the modelling and analysis of a wide range of open quantum systems. In QSDEs, the environment is modelled by external fields acting on a boson Fock space \cite{P_1992}. In particular, linear QSDEs represent the Heisenberg evolution of pairs of conjugate operators in a multi-mode open quantum harmonic oscillator (OQHO) which is coupled to external bosonic fields. For example, in quantum optics, the input-output dynamics of quantum-optical components, such as optical cavities, beam splitters and phase shifters, and their interconnections are often modelled by linear QSDEs \cite{P_1992,P_2010}, provided the latter are physically realizable (PR) as OQHOs \cite{GZ_2004}. The conditions for PR of linear QSDEs are organised as a set of constraints on the coefficients of the QSDEs \cite{JNP_2008} or, alternatively, on the quantum system transfer function in the frequency domain \cite{SP_2012}. 

In linear feedback control systems, it is the transfer function of the controller, and not the particular state-space realization of the controller, which determines the important specifications of the closed-loop system such as stability. Similarly, in coherent quantum feedback control problems, in which the controller is required to be PR (see for example \cite{JNP_2008}), it is important to have a condition for physically realizability on the controller transfer function. This condition can be used to give a unifying treatment for coherent quantum synthesis problems requiring stability and high performance in terms of $H_2$ and $H_\infty$ norms \cite{AKh_2015}. Moreover, the condition can be utilized to facilitate the application of frequency domain approaches to model approximation, reduction and system identification of linear quantum systems (see for example \cite{SP_2012, P_2010, MY_2016} and references therein). These applications motivate the study of PR conditions on quantum system transfer functions which are shown to be equivalent to a frequency domain $(J,J)$-unitary constraint and a unitary symplectic constraint on the direct feedthrough of the quantum system, under some technical assumptions, in \cite{SP_2012}. 

In the present note, we provide a relatively simple proof to a modified version of the results of \cite{SP_2012} which avoids the technical assumptions required in that paper. In view of these new results, associated coherent control problems for linear quantum systems can be addressed by purely frequency domain approaches. Indeed, removing the technical assumptions from the results of \cite{SP_2012} is important, because it makes the application of the result, for example, in coherent quantum control, simpler and more complete since the technical assumption does not need to be checked. Moreover, we provide a connection between the location and number of poles and transmission zeros associated with the transfer functions of PR linear quantum systems. In particular, we show that the transmission zeros of such transfer functions are the mirror reflections of its poles about the imaginary axis. Finally, we provide an example to illustrate the results. 

The rest of this paper is organised as follows. Section~\ref{sec:not} outlines the notation used in the paper. We provide a brief introduction to the OQHOs under consideration in Section~\ref{sec:system}. Section~\ref{sec:LQHO_PR} describes the PR condition for the quantum systems and provides some facts about the location of the poles and zeros of their transfer functions. Finally, we provide an example to illustrate the results of the paper. Some additional results required in the proofs are given in appendices~\ref{App:1_1_C} and \ref{App:CLD}.
\section{Notation}\label{sec:not}
Unless specified otherwise,  vectors are organized as columns, and the transpose $(\cdot)^{\rT}$ acts on matrices with operator-valued entries as if the latter were scalars. For a vector $X$ of self-adjoint operators $X_1, \ldots, X_r$ and a vector $Y$ of operators $Y_1, \ldots, Y_s$, the commutator matrix is defined as an $(r\x s)$-matrix
$
    [X,Y^{\rT}]
    :=
    XY^{\rT} - (YX^{\rT})^{\rT}
$
whose $(j,k)$th entry is the commutator
$
    [X_j,Y_k]
    :=
    X_jY_k - Y_kX_j
$ of the operators $X_j$ and $Y_k$. Furthermore, $(\cdot)^{\dagger}:= ((\cdot)^{\#})^{\rT}$ denotes the transpose of the entry-wise operator adjoint $(\cdot)^{\#}$. When it is applied to complex matrices,  $(\cdot)^{\dagger}$ reduces to the complex conjugate transpose  $(\cdot)^*:= (\overline{(\cdot)})^{\rT}$. 
$\Re M$ and $\Im M$ denote the extension of the real and imaginary part of a complex matrix to matrices $M$ with operator-valued entries as $\Re M = \frac{1}{2} (M+M^\#)$ and $\Im M = \frac{1}{2i} (M-M^\#)$ which consist of self-adjoint operators. The positive semi-definiteness of matrices is denoted by $\succcurlyeq$, and $\ox$ is the tensor product of spaces or operators (for example, the Kronecker product of matrices). Furthermore, $\mS_r$, $\mA_r$
 and
$
    \mH_r
    :=
    \mS_r + i \mA_r
$ denote
the subspaces of real symmetric, real antisymmetric and complex Hermitian  matrices of order $r$, respectively, with $i:= \sqrt{-1}$ the imaginary unit. Also, $I_r$ denotes the identity matrix of order $r$, 
$
J_r:=
{\scriptsize\begin{bmatrix}
     0 & 1\\
    -1 & 0
\end{bmatrix}}
\ox
I_{\frac{r}{2}}
$ and
$
 \bJ_{r}:=
{\scriptsize\begin{bmatrix}
     1 & 0\\
     0 & -1
\end{bmatrix}}
\ox
I_{\frac{r}{2}}     
$. 
The sets $O(2r) := \big\{\Sigma \in \mR^{2r\x 2r}: \Sigma^\rT \Sigma =I \big\}$ and $Sp(2r,\mR) := \big\{\Sigma \in \mR^{2r\x 2r}: \Sigma^\rT J_{2r} \Sigma =J_{2r} \big\}$ refer to the group of orthogonal matrices and the group of symplectic real matrices of order $2r$. Matrices of the form
${\scriptsize\begin{bmatrix}
    R_1 & R_2 \\
    \overline{R}_2 & \overline{R}_1
\end{bmatrix}}
$
are denoted by $\Delta(R_1,R_2)$.
The notation
$
    {\scriptsize\left[
    \begin{array}{c|c}
          A & B \\
          \hline
          C & D
    \end{array}
    \right]}
$
refers to a state-space realization of the corresponding transfer matrix $\Gamma(s) := C(sI-A)^{-1}B+D$ with a complex variable $s \in \mC$.
The conjugate system transfer function $(\Gamma(-\overline{s}))^*$ is written as $\Gamma^{\~}(s)$.

\section{Open Quantum Harmonic Oscillators} \label{sec:system}
We consider the joint evolution of an $n$-mode OQHO and external bosonic fields in the Heisenberg picture, represented by the linear QSDEs:
\vskip-5mm\begin{align}
\label{equ:tdomain_model:1}
    \rd
    X(t)
    &= A X(t) \rd t+
       B \rd W(t),\\
\label{equ:tdomain_model:2}
    \rd
    Y(t)
    &= C X(t)\rd t+
   			 D \rd W(t).
\end{align}\noindent
Here, the first QSDE governs the plant dynamics, while the second QSDE describes the dynamics of the output fields on the system-field composite Hilbert space $\cH \ox \cF$. The vector $X$ of dynamic variables  satisfies the canonical commutation relations (CCRs)
\vskip-2mm
\begin{equation*}
    [X, X^{\rT}] = 2i \Theta,
    \qquad
    X:=
    {\small\begin{bmatrix}
        X_1\\
        \vdots\\
        X_{2n}
    \end{bmatrix}}
\end{equation*}
with a non-singular CCR matrix $\Theta \in \mA_{2n}$. Also, $W$ is a $2m$-dimensional vector of quantum Wiener processes $W_{1}, \ldots, W_{2m}$, which are self-adjoint operators on a boson Fock space \cite{H_1991,P_1992}, modelling the external fields with the It\^{o} matrix $\Omega:=\big( \omega_{jk} \big)_{1\<j,k\<2m} \in \mH_{2m}$:
\begin{equation}
\label{WW}
    \rd W \rd W^{\rT}
    =
    \Omega \rd t.
\end{equation}  
The entries of $W$ are linear combinations of the field annihilation $\fA_1, \hdots, \fA_m$ and creation $\fA_1^\dagger, \hdots, \fA_m^\dagger$ operator processes \cite{HP_1984,P_1992}:
\begin{equation}
\label{QWP}
{\scriptsize
	W\!\!:=\!\!2	
	\begin{bmatrix}
		\Re  \fA\\
		\Im  \fA
	\end{bmatrix}		
	\!\!=\!\!
	T_{2m}
	\begin{bmatrix}
		\fA\\
		\fA^\#
	\end{bmatrix}\!\!,\qquad
	T_{2m}
	\!\!:=\!\!	
		\begin{bmatrix}
			1  & 1\\
			-i & i
		\end{bmatrix}	
		\ox
	I_m	
	.
}
\end{equation}
The field annihilation and creation operators are adapted to the Fock filtration with the quantum It\^{o} relations
\begin{equation*}	
	\small
	\rd 
	\begin{bmatrix}
		\fA\\
		\fA^\#
	\end{bmatrix}
	\rd
	\begin{bmatrix}
		\fA^\dagger 
		&
		\fA^\rT
	\end{bmatrix}	
	\!\!:=\!\!
	\begin{bmatrix}
		\rd \fA \rd \fA^\dagger	   &	\rd \fA \rd \fA^\rT\\
		\rd \fA^\# \rd \fA^\dagger &    \rd \fA^\# \rd \fA^\rT
	\end{bmatrix}
	\!\!=\!\!
	\bigg(
	\begin{bmatrix}
		1 & 0\\
		0 & 0
	\end{bmatrix}
	\ox I_m
	\bigg)
	\rd t.
\end{equation*}
Accordingly, the It\^{o} matrix $\Omega$ in (\ref{WW}) is described by
\begin{equation}
	\label{Omega}
	\small
	\Omega=
	\bigg( 
	\begin{bmatrix}
		1 & 1\\-i & i
	\end{bmatrix}
	\begin{bmatrix}
		1 & 0\\0 & 0 
	\end{bmatrix}
	\begin{bmatrix}
		1 & 1\\-i & i 
	\end{bmatrix}^*
	\bigg)
	\ox I_m
	=
	I_{2m}+iJ_{2m} \succeq 0.
\end{equation}
In what follows, the subscripts in $I_{2m}$ and $J_{2m}$ will often be omitted for brevity.
The matrices $A \in \mR^{2n\x 2n}$, $B \in \mR^{2n\x 2m}$, $C \in \mR^{2m\x 2n}$, $D \in \mR^{2m\x 2m}$ in (\ref{equ:tdomain_model:1}) and (\ref{equ:tdomain_model:2}) are given by
\vskip-3mm\begin{align}
    \label{equ:ABCD}
    {\small\begin{bmatrix}
        A & B\\
        C & D
    \end{bmatrix}}
    := 
    {\small\begin{bmatrix}
        2\Theta R - \frac{1}{2} BJB^{\rT} \Theta^{-1} & B\\
        -D J B^\rT \Theta^{-1} & D
    \end{bmatrix}}, \quad
	B:=2 \Theta M^\rT.
\end{align}\noindent
Also, the parameter $R$ is a real symmetric matrix of order $2n$ associated with the quadratic Hamiltonian $\frac{1}{2} X^\rT R X$ of the OQHO, the linear system-field coupling parameter $M\in \mR^{2m\times 2m}$ and, in view of a similar relation in (\ref{WW}) for the output fields, the feedthrough real matrix $D$ belongs to the subgroup of orthogonal symplectic matrices (the maximum compact subgroup of symplectic matrices) 
\begin{equation}
	\label{Spm}
	Sp(m) = O(2m) \cap Sp(2m, \mR).
\end{equation}
Note that there exists a one-to-one correspondence between the real-valued parameterization (\ref{equ:ABCD}) with independent parameters $D$, $M$, $R$, which will be referred to as the position-momentum form of OQHOs, and the complex-valued, but structured, parameterization, referred to as the annihilation-creation form of OQHOs \cite{P_2010}; see Appendix~\ref{App:1_1_C} for more details. In \cite{SP_2012}, use is made of the annihilation-creation form of OQHOs to address the PR conditions for quantum systems.
\section{Open Quantum Harmonic Oscillators in the Frequency Domain and Physical Realizability}\label{sec:LQHO_PR} 
The input-output map of the OQHO, governed by the linear QSDEs (\ref{equ:tdomain_model:1}) and (\ref{equ:tdomain_model:2}), is completely specified by a transfer function which is defined in the standard way as
    \vskip-2mm\begin{equation}
    \label{equ:hoc_freq}
    \Gamma(s) :=
    {\small\left[
    \begin{array}{c|c}
          A & B\\
          \hline
          C & D
    \end{array}
    \right]}=C(sI-A)^{-1}B+D,
    \end{equation}\vskip-1mm\noindent
where the matrices $A, B, C, D$ are parameterized by the triplet $(D,M,R)$ as in (\ref{equ:ABCD}) with a given CCR matrix $\Theta$. In view of the specific structure of this parameterization, not every linear system, or system transfer function (\ref{equ:hoc_freq}) with an arbitrary quadruple $(A,B,C,D)$, represents the dynamics of an OQHO. This fact is addressed in the form of PR conditions for the quadruple $(A,B,C,D)$ to represent such an oscillator; see \cite{JNP_2008} for more details. The notion of PR for a transfer function is defined as follows.
\begin{definition} \label{def:PR_TF}
	The transfer function $\Gamma(s)$ is said to be \emph{physically realizable} if $\Gamma(s)$ represents an 
	OQHO, that is, there exists a minimal state-space realization for $\Gamma(s)$ which can be parameterized 
	by a triplet $(D,M,R)$ as in (\ref{equ:ABCD}) for a given CCR matrix $\Theta$.
\end{definition}

Note that, in view of the results of Lemma~\ref{lem:ch_fact} in Appendix~\ref{App:Ch_dec}, invariance of transfer functions with respect to similarity transformations of their state-space realizations \cite{ZDG_1996} and Definition~\ref{def:PR_TF}, by a similar approach which will be used in (\ref{eq:TF_DMR}), it can be shown that $\Gamma(s)$ is also physically realizable if there exists a minimal state-space realization for $\Gamma(s)$ which can be parameterized by the triplet $(D,M,R)$ as in (\ref{equ:ABCD}) with any non-singular skew-symmetric matrix $\Theta$. The following theorem which is the main result of this paper provides a PR condition for transfer matrices of linear quantum systems, which can be considered as a modified version of Theorem~4 in \cite{SP_2012}.
\begin{thm}
\label{thm:PR_Freq}
A transfer function $\Gamma(s)$ is physically realizable if and only if
	\vskip-3mm\begin{equation}
		\label{equ:JJUnit}
		\Gamma^{\~}(s) J \Gamma(s)=J
	\end{equation}\vskip-1mm\noindent
for all $s\in \mC$, and the feedthrough matrix $D = \Gamma(\infty)$ is orthogonal. 
\end{thm}
\begin{proof}
	By assuming that (\ref{equ:JJUnit}) is satisfied for all $s\in\mC$, the feedthrough matrix $D$ inherits the symplectic property, that is $D \in Sp(2m,\mR)$, from the transfer function $\Gamma(s)$ by continuity. Then, since the feedthrough matrix $D \in O(2m)$, we have $D \in Sp(m)$, where $Sp(m)$ is given in (\ref{Spm}). Moreover, the inverse of $\Gamma(s)$ can be computed as 
	\begin{equation}
		\label{eq:pr1:PR}
		\Gamma^{-1}(s)=- J \Gamma^{\~}(s) J .
	\end{equation}
	Since $\Gamma(s)$ is a proper transfer function, there exists minimal state-space realization for $\Gamma(s)$. By considering  (\ref{equ:hoc_freq}) as a minimal realization of $\Gamma(s)$, a minimal realization for the inverse transfer function is given by
	\begin{equation*}
		\Gamma^{-1}(s) =
	    {\small\left[
	    \begin{array}{c|c}
	          A-BD^{-1}C & BD^{-1}\\
	          \hline
	          -D^{-1}C & D^{-1}
	    \end{array}
	    \right]},
	\end{equation*}
	(see \cite[proposition~4.1.5]{B90}). In view of (\ref{eq:pr1:PR}),
	\begin{align}
		\nonumber
		D^{-1} - D^{-1} C \big( s I - A + BD^{-1}C \big)^{-1} BD^{-1}
		=\\ \label{eq:pr2:PR}
		-J \big( D^\rT - B^{\rT} (s I+A^\rT)^{-1}C^\rT \big) J,
	\end{align}
	which is an equality between two minimal realizations of the same rational transfer function matrix. Then, there exists a unique real and invertible  matrix $F$, associated with a state-space similarity transformation, (see, for example, \cite[Theorem~3.17]{ZDG_1996}) such that  
	\begin{align}
		\label{eq:pr3:PR}
		J B^\rT F \!=\! -D^{-1}C, \ \ 
		F^{-1}C^\rT J \!=\! BD^{-1}, \ \ 
		-F^{-1}A^\rT F \!=\! A - BD^{-1}C.
	\end{align}
	By transposing and rearranging the equations in (\ref{eq:pr3:PR}), and using the fact that $D^\rT J D =J$, we see that $-F^\rT$ also satisfies these equations. Therefore, from the uniqueness of $F$, it follows that $F=-F^\rT$.  Moreover, it can be shown by inspection from these equations that
\begin{align}
	\label{eq:pr4:PR}
	 C &= -DJ B^\rT F,\\	
	\label{eq:pr6:PR}
	0 &= A^\rT F^\rT + F^\rT A + C^{\rT} J C,\\
	\label{eq:pr5:PR}
	0 &= AF^{-1} + F^{-1} A^\rT+BJ B^{\rT}.
\end{align}	
	Equation (\ref{eq:pr5:PR}) implies $A = 2F^{-1} \wh{R} - \frac{1}{2} B J B^\rT F$ for 
	\begin{equation}
		\label{Rhat}
		\wh{R} := \frac{1}{2} F \Big( AF^{-1}+\frac{1}{2} B J B^\rT \Big) F = \wh{R}^\rT.
	\end{equation}
	In view of the results of Lemma~\ref{lem:ch_fact} in Appendix~\ref{App:Ch_dec} and the fact that any non-singular skew-symmetric matrix, such as $F$, is necessarily of even order, there exists a non-singular matrix $\Sigma \in \mR^{2n \times 2n}$ such that $F^{-1} = \Sigma \Theta \Sigma^{\rT}$ for any given CCR matrix $\Theta \in \mA^{2n \times 2n}$. Then, the $(D,M,R)$ parameters for the transfer function $\Gamma(s)$ can be represented as
\begin{equation}
	\label{eq:TF_DMR}
	( D , -\frac{1}{2} B^\rT \Sigma^{-\rT} \Theta^{-1} , \Sigma^{\rT} \wh{R} \Sigma),
\end{equation}
where $\wh{R}$ is defined in (\ref{Rhat}). Hence, $\Gamma(s)$ is physically realizable.

Conversely, suppose the transfer function (\ref{equ:hoc_freq}) is physically realizable and hence there exists a triplet $(D,M,R)$ such that (\ref{equ:ABCD}) holds. We compute
\begin{align*}
	\Gamma^{\~}(s) J \Gamma(s) 
	&= 
	\big(
		D^\rT - B^\rT (s I+A^\rT)^{-1}C^\rT
	\big)	
	J
	\big(
		D+C(s I-A)^{-1}B
	\big)
	\\
	&= 
	D^\rT J D + D^\rT J C(s I - A)^{-1} B - B^{\rT}(s I + A^\rT)^{-1}C^\rT J D  \\
	&\quad - B^\rT(s I+A^\rT)^{-1}C^\rT J C(sI-A)^{-1} B.
\end{align*}
It can be shown by inspection that similar equations to (\ref{eq:pr4:PR}) and (\ref{eq:pr6:PR}) with $F=\Theta^{-1}$ are satisfied for the realization $(A,B,C,D)$. Then, by replacing $C^\rT J D$ with $F B$ and $C^\rT J C$ with $A^\rT F + F A$ and using $D^\rT J D=J$ we obtain
\begin{small}
\begin{align*}
	\Gamma^{\~}(s) J \Gamma(s) 
	&=
	D^\rT J D + D^\rT J C(s I - A)^{-1} B - B^{\rT}(s I + A^\rT)^{-1}C^\rT J D  \\
	&\quad - B^\rT(s I+A^\rT)^{-1}C^\rT J C(sI-A)^{-1} B\\
	&= 
	J  + B^\rT F^\rT (s I - A)^{-1} B 	- B^{\rT}(s I + A^\rT)^{-1}FB \\
	&\quad - B^\rT(s I+A^\rT)^{-1} \big( A^\rT F + F A  + sF-Fs \big)(sI-A)^{-1} B\\
	&= 
	J  + B^\rT F^\rT (s I - A)^{-1} B 	- B^{\rT}(s I + A^\rT)^{-1}FB \\
	&\quad - B^\rT F (sI-A)^{-1} B	+ B^\rT(s I+A^\rT)^{-1}F B\\
	&=J
\end{align*}
\end{small}\noindent
where use is made of the skew-symmetry of $F$. This implies that $\Gamma(s)$ satisfies (\ref{equ:JJUnit}) for all $s\in\mC$.
\end{proof}
A transfer function $\Gamma(s)$, satisfying the condition (\ref{equ:JJUnit}), is said to be $(J,J)$-unitary; see, for example, \cite{SP_2012} and references therein. Since we consider this property for invertible square transfer matrices, in view of the fact that $J^2=-I$, the $(J,J)$-unitarity is equivalent to its dual form \cite{AKh_2015}:
\vskip-2mm
\begin{equation*}
		\Gamma(s) J \Gamma^{\~}(s)=J.
\end{equation*}
	
	In view of the one-to-one correspondence described in Appendix~\ref{App:1_1_C}, the results in Theorem~\ref{thm:PR_Freq} imply the results in \cite[Theorem~4]{SP_2012}. In particular, in the annihilation-creation form of OQHOs a similar result to Theorem~\ref{thm:PR_Freq} can be derived by replacing the matrix $J$ with $\bJ$ and $\Gamma(s)$ with 
$\pmb{\Gamma}(s) :=
    {\scriptsize\left[
    \begin{array}{c|c}
          F & G \\
          \hline
          L & K
    \end{array}
    \right]}
$, where the quadruple $(F,G,L,K)$ are defined in (\ref{equ:FGLK}).  Also, $K = \pmb{\Gamma}(\infty)$ must be of the form $\Delta(S,0)$ in which $S$ is a unitary matrix. However, in comparison to \cite[Theorem~4]{SP_2012}, no additional technical assumptions are required in Theorem~\ref{thm:PR_Freq}. The technical assumption which is used in \cite{SP_2012} is referred to as spectral genericity of the linear quantum systems \cite{AKh_2015}; refer to Definition~\ref{def:gen} and the corresponding definition in the position-momentum form of OQHOs in Appendix~\ref{App:1_1_C}.

In what follows, the notion of transmission zeros will be used according to their standard definition in linear systems theory; see for example \cite{ZDG_1996}.
\begin{corollary}
	\label{cor:QHO_PZ}
	Consider an OQHO with associated transfer function $\Gamma(s)$. The transmission zeros of $\Gamma(s)$ are the mirror reflection about the imaginary axis of its poles. 
\end{corollary}
\begin{proof}
In view of the results of Theorem~\ref{thm:PR_Freq}, as shown in (\ref{eq:pr3:PR}), the existence of a non-singular $F \in \mA_{2n}$ such that
\begin{equation*}
	-F^{-1}A^\rT F \!=\! A - BD^{-1}C
\end{equation*}  
implies that the spectrum $\sigma\big(\!\!-A^\rT \big)$ coincides with the spectrum $\sigma\big(A - BD^{-1}C \big)$ where the former coincides with the mirror reflection about the imaginary axis of the poles (the eigenvalues of the real matrix $A$) and the latter coincides with the transmission zeros of the transfer function $\Gamma(s)$ \cite{ZDG_1996}. 
\end{proof}
\section{Illustrative Example}\label{sec:exmpl}
\begin{exmp}
Consider a transfer matrix 
\begin{small}
\begin{equation*}
\Gamma(s) = \diag \bigg(\frac{s+1}{s},\frac{s-1}{s+1},\frac{s}{s-1},\frac{s-1}{s+1}\bigg)
\end{equation*}
\end{small}\noindent
which satisfies the conditions of Theorem~\ref{thm:PR_Freq}, that is, $\Gamma^{\~}(s) J \Gamma(s) = J$ for all $s \in \mC$ and $\Gamma(\infty) \in O(4)$. Then the transfer function $\Gamma(s)$ represents an OQHO.
The parameters $D$, $M$, $R$ for the associated OQHO with $\Theta = J$ are given by
\begin{small}
\begin{align*}
	D = I,
	\quad
	R=
	\begin{bmatrix}
		0 & 0 & \frac{1}{4} & 0 \\	0 & 0 & 0 & 0\\ \frac{1}{4} & 0 & 0 & 0\\ 0 & 0 & 0 & 0
	\end{bmatrix}\!\!\!,
	\quad	
	M=
	\begin{bmatrix}	
		-\frac{1}{2} & 0 & 0 & 0\\0 & 0 & 0 & 1\\0 & 0 & \frac{1}{4}  & 0\\ 0 & -\frac{1}{2} & 0 & 0
	\end{bmatrix}\!\!\!.			
\end{align*}	
Also, in view of the one-to-one correspondence between OQHOs in the position-momentum form and OQHOs in the annihilation-creation form, the complex-valued parameters $\bS$ ,$\bH$, $\bN$ with $\pmb{\Theta}= \bJ$ are given by
\begin{align*}
	\bS = I,\!\!
	\quad
	\bH =
	\begin{bmatrix}	
		0 & 0 & \frac{i}{2} & 0 \\	0 & 0 & 0 & 0\\ -\frac{i}{2} & 0 & 0 & 0\\ 0 & 0 & 0 & 0
	\end{bmatrix}\!\!\!,	\!\!\!\!	
	\quad
	\bN=
	\begin{bmatrix}	
		  0 & 0 & i & 0\\0 & -\frac{3}{2} & 0 & \frac{1}{2}\\ - i & 0 & 0 & 0 \\ 0 & \frac{1}{2} & 0 & -\frac{3}{2}
	\end{bmatrix}\!\!\!.
\end{align*} 
\end{small}
\vskip-2mm
The corresponding transfer matrix in the annihilation-creation form is
\begin{small}
\begin{equation*}
	\pmb{\Gamma}(s) =
	\begin{bmatrix}
		\frac{s^2-\frac{1}{2}}{s(s-1)} & 0 & \frac{-1}{2s(s-1)} & 0\\
 		0 & \frac{s-1}{s+1} & 0 & 0 \\
 		\frac{-1}{2s(s-1)} & 0 & \frac{s^2-\frac{1}{2}}{s(s-1)} & 0\\
		0 & 0 & 0 & \frac{s-1}{s+1}
	\end{bmatrix},
\end{equation*} 
\end{small}
and its associated McMillan form \cite{ZDG_1996} is
\begin{small}
\begin{equation*}
	\bM(s)=\diag \bigg(\frac{1}{s^3-s},\frac{1}{s+1},s-1,s^3-s\bigg).
\end{equation*}
\end{small}
The poles of $\Gamma(s)$ (and $\pmb{\Gamma}(s)$) are $(0,-1,-1,1)$, and hence, according to Definition~\ref{def:gen}, there exists no spectrally generic realization for $\Gamma(s)$ (or $\pmb{\Gamma}(s)$). Therefore, the results of \cite{SP_2012} cannot be applied to this example. The transmission zeros of $\Gamma(s)$ are $(0,1,1,-1)$, which are the mirror reflection about the imaginary axis of the poles.
\end{exmp}
\section*{Acknowledgment}
	AKhS is grateful to Dr. Igor G. Vladimirov for useful discussions and comments on this work.
\section{Conclusion}
We have shown that the PR condition is equivalent to a $(J,J)$-unitarity constraint on the quantum system transfer function and an orthogonality constraint on the constant feedthrough of the system. The technical assumption on existence of a spectrally generic realization of the transfer function associated with OQHOs used in the previous results has been shown to be redundant and a relatively simple proof has been provided to validate the modified results. We have also shown that the poles and transmission zeros, associated with the transfer functions of linear quantum systems, are the mirror reflection about the imaginary axis of each other. 
\appendix
\renewcommand{\theequation}{\Alph{subsection}\arabic{equation}}
\vskip-1mm
\subsection{One-to-One Correspondence Between Annihilation-Creation and Position-Momentum Forms of Open Quantum Harmonic Oscillators} \label{App:1_1_C}
\setcounter{equation}{0}
In order to make a connection between the results of Section~\ref{sec:LQHO_PR} and the results of \cite{SP_2012}, this section provides a one-to-one correspondence between the annihilation-creation and position-momentum forms of OQHOs.

Corresponding to a model of $n$ independent OQHOs is a vector $\fa$ of annihilation operators $\fa_1, \ldots, \fa_n$ on Hilbert spaces $\cH_1, \ldots, \cH_n$. The adjoint $\fa_j^\dagger$ of the operator $\fa_j$ is referred to as the creation operator. The doubled-up vector $\breve{\fa}$ of the annihilation and creation operators satisfies the CCRs \cite{M_1998}
\begin{equation}
\label{aaa}
    [
        \breve{\fa},
        \breve{\fa}^{\dagger}
    ]
    \!\!:=\!\!
    {\small\begin{bmatrix}
        [\fa,\fa^{\dagger}] & [\fa,\fa^{\rT}]\\
        [\fa^{\#},\fa^{\dagger}] & [\fa^{\#},\fa^{\rT}]
    \end{bmatrix}}
    \!\!= \!\!
    \bJ_{2n}
    ,\quad   
    \breve{\fa}
    \!\!:= \!\!
    {\small\begin{bmatrix}
        \fa\\ \fa^{\#}\end{bmatrix}}.
\end{equation}\noindent
We consider a linear quantum system whose dynamic variables are linear combinations of the annihilation and creation operators, acting on the tensor product space $\cH:= \cH_1\ox \ldots \ox \cH_n$:
\begin{equation}
\label{EEa}
    a
    :=
    \bE_1 \fa+ \bE_2 \fa^{\# }
    =
    \begin{bmatrix}
        \bE_1 & \bE_2
    \end{bmatrix}
    \breve{\fa},
\end{equation}
where $\bE_1$ and $\bE_2$ are appropriately dimensioned complex matrices.
The relations (\ref{aaa}) and (\ref{EEa}) imply that
\vskip-2mm$$
    [\breve{a}, \breve{a}^{\dagger}]
    =
    \bE
    [\breve{\fa}, \breve{\fa}^{\dagger}]
    \bE^*
    =
    \bE \bJ_{2n} \bE^*
    =:
    \pmb{\Theta},
$$\noindent
where $\bE:= \Delta(\bE_1,\bE_2) \in \mC^{2n \times 2n}$ is a non-singular matrix in accordance with the doubled-up notation, and the complex Hermitian matrix $\pmb{\Theta}$ of order $2n$ is the (generalized) CCR  matrix \cite{SP_2012}.
Now, consider an $n$-mode OQHO interacting with an external bosonic field defined on a Fock space \cite{P_1992}. The oscillator is assumed to be coupled to $m$ independent external input bosonic fields acting on the tensor product space $\cF:=\cF_1 \ox\ldots \ox\cF_m$, where $\cF_j$ denotes the Fock space associated with the $j$th input channel. The field annihilation operators $\fA_1(t), \ldots, \fA_m(t)$, which act on $\cF$, form a vector $\fA(t)$. Their adjoints  $\fA_1^{\dagger}(t), \ldots, \fA_m^{\dagger}(t)$, that is, the field creation operators, comprise a vector $\fA^\#(t)$. The field annihilation and creation operators are adapted to the Fock filtration and satisfy the It\^{o} relations
$
	\scriptsize
    \rd
    \breve{\fA}(t)
    \rd
    \breve{\fA}^{\dagger}(t)
    =
    {\begin{bmatrix}
        I_m & 0\\
        0 & 0
    \end{bmatrix}}
    \rd t
$ in terms of the corresponding doubled-up vector 
$\scriptsize \breve{\fA}(t):=   
	{\begin{bmatrix}
        \fA(t)\\
        \fA^{\#}(t)
    \end{bmatrix}}$. 
The linear QSDEs, derived from the joint evolution of the $n$-mode OQHO and the external bosonic fields in the Heisenberg picture, can be represented in the following form \cite{P_2010,SP_2012}:
\vskip-5mm\begin{align}
\label{equ:tdomain_A_model:1}
    \rd
    \breve{a}(t)
    &= F \breve{a}(t) \rd t+
       G \rd \breve{\fA}(t),\\
\label{equ:tdomain_A_model:2}
    \rd
    \breve{\fA}_{\rm out}(t)
    &= L \breve{a}(t)\rd t+
   	   K \rd \breve{\fA}(t).
\end{align}\noindent
Here, the first QSDE governs the plant dynamics, while the second QSDE describes the dynamics of the output fields in terms of the corresponding doubled-up vector
$
\scriptsize
\breve{\fA}_{\rm out}(t)
    :=
    { \begin{bmatrix}
        \fA_{\rm out}(t)\\
        \fA_{\rm out}^{\#}(t)
    \end{bmatrix}}
$
of annihilation and creation operators acting on the system-field composite   space $\fH\ox \fF$. Also, the matrices $F \in \mC^{2n\x 2n}$, $G \in \mC^{2n\x 2m}$, $L \in \mC^{2m\x 2n}$, $K \in \mC^{2m\x 2m}$ in (\ref{equ:tdomain_model:1}) and (\ref{equ:tdomain_model:2}) are given by
\vskip-3mm\begin{align}
    \label{equ:FGLK}
    {\small\begin{bmatrix}
        F & G\\
        L & K
    \end{bmatrix}}
    := 
    {\small\begin{bmatrix}
        -i\pmb{\Theta} \bH - \frac{1}{2} \pmb{\Theta} \bN^*\bJ_{2m} \bN & -\pmb{\Theta} \bN^* \bJ_{2m} 
        \Delta(\bS,0)\\
        \bN & \Delta(\bS,0)
    \end{bmatrix}},
%
\end{align}\noindent
where $\bH=\bH^*=\Delta(\bH_1,\bH_2) \in \mC^{2n \times 2n}$ is a Hermitian matrix which parameterizes the system Hamiltonian operator $\frac{1}{2}\breve{a}^{\dagger} \bH \breve{a}$, the matrix $\bN:=\Delta(\bN_1,\bN_2) \in \mC^{2m \times 2n}$ specifies the system-field coupling operators, and $\bS \in \mC^{m \times m}$ is the unitary scattering matrix. 

Similarly to (\ref{QWP}) and (\ref{Omega}), we define
{\small
\begin{align}
	\label{QXYP}
	X \!\!  :=\!\! 2
	\begin{bmatrix}
		\Re  a\\
		\Im  a
	\end{bmatrix}
	=
	T_{2n}\breve{a}
	,\qquad
	Y \!\! :=\!\! 2
	\begin{bmatrix}
		\Re  \fA_{\rm out}\\
		\Im  \fA_{\rm out}
	\end{bmatrix}
	= T_{2m} \breve{\fA}_{\rm out}
\end{align}
}\noindent
which provides a one-to-one correspondence between the OQHOs in the annihilation-creation form, parameterized by the matrices $\bS$, $\bN$, $\bH$ in (\ref{equ:tdomain_A_model:1}), (\ref{equ:tdomain_A_model:2}), and the OQHOs in the position-momentum form, parameterized by the matrices $D$, $M$, $R$ in (\ref{equ:tdomain_model:1}), (\ref{equ:tdomain_model:2}):
\begin{align}
	\label{DS}
	D &= 
		\nabla(\bS,0),
	\\ \label{MS}
	M &= -\frac{1}{2} \nabla^\rT(\bS,0) J_{2m} \nabla(\bN_1,\bN_2),
	\\ \label{RH}
	R
	&=
	\frac{1}{2} \nabla(\bH_1,\bH_2),
	\\ \label{tT}
	\Theta
	&=
	\nabla(\bE_1,\bE_2) J_{2n} \nabla(\bE_1,\bE_2)^\rT,
\end{align}
where we define the real matrix-valued function $\nabla(\bX_1,\bX_2) \in \mR^{2k \x 2j}$ for given matrices $\bX_\ell \in \mC^{k \times j}$ (such as $\bN_\ell$, $\bH_\ell$, $\bE_\ell$) for $\ell=1,2$ as 
\begin{align}
	\label{NbLa}
	\scriptsize
	\nabla(\bX_1,\bX_2) 
	:= 
	\frac{1}{2} T_{2k} \Delta(\bX_1,\bX_2) T^*_{2j} 
	=
	\begin{bmatrix}
		\Re (\bX_1 +  \bX_2)	&   -\Im (\bX_1 - \bX_2)\\
		\Im (\bX_1 +  \bX_2)   &	 \Re (\bX_1 - \bX_2)
	\end{bmatrix}.
\end{align}
Also, use is made of $T_k T_k^*=T_k^*T_k=2I_k$ and $\frac{1}{2} T_{2k} \bJ_{2k} T_{2k}^* = iJ_{2k}$ in (\ref{DS})--(\ref{tT}).
It follows from (\ref{DS}), (\ref{HR}), (\ref{NbLa}) and the Hermitian property of $\bH$ that
\begin{equation*}	
	D^\rT(I_{2m}+iJ_{2m})D = I_{2m}+iJ_{2m},\qquad R=R^\rT.
\end{equation*}
Conversely, for given parameters $(D,M,R)$ of OQHOs in the position-momentum form
\begin{align}
	\label{SD} 
	K 
	&= 
	\Delta(
		\bD_1
	   	,
	   	\bD_2
	   	)
	= 
	\Delta(
		\bD_1
	   	,
	   	0),
	\\	
	\label{LM}
	\bN 	
	&= 
	-2i\Delta(\bD_1,0) \bJ_{2m} \Delta(\bM_1,\bM_2),
	\\ \label{HR}
	\bH &= 2 \Delta(
	   			\bR_1
	   			,
	   			\bR_2
	   			),	
	\\ \label{bThTh}
	\pmb{\Theta}
	  &= \Delta(
	   			\bE_1
	   			,
	   			\bE_2
	   	 )
	   	 \bJ_{2n}
	   	 \Delta(
	   			\bE_1
	   			,
	   			\bE_2
	   	 )^*
	   	 ,	   	 		   			   			
\end{align}
where we partition $(2j\x 2k)$-matrices $X$ (such as $D$, $M$, $R$, $E$) into $(j\x k)$-blocks as
\begin{equation*}	
    X
    :=
     \begin{bmatrix}
                        X_{11} & X_{12}\\
                        X_{21} & X_{22}
    \end{bmatrix},    
\end{equation*}
and $\bX_1$ and $\bX_2$ are defined as
\begin{align*}
	\bX_1&:= \frac{1}{2}(X_{11}+X_{22})+\frac{i}{2}(X_{21}-X_{12}),\\
	\bX_2&:= \frac{1}{2}(X_{11}-X_{22})+\frac{i}{2}(X_{21}+X_{12})
\end{align*}
and $E$ can be computed from a Cholesky-like factorization as $\Theta=EJ_{2n}E^\rT$; refer to Appendix~\ref{App:Ch_dec}.
Also, use is made of the fact that $D \in Sp(m)$, where $Sp(m)$ is defined in (\ref{Spm}), in (\ref{SD}) which implies $\bD_1^* \bD_1=\bD_1 \bD_1^*=I$ and $\bD_2^* \bD_2=0$. Then, $\bS= \bD_1$ is a unitary matrix and $\bD_2=0$. It follows from the symmetric property of $R$ and non-singularity of $E$ that $\bH$ and $\pmb{\Theta}$, defined in (\ref{HR}), (\ref{bThTh}), are Hermitian matrices and $\pmb{\Theta}$ is a non-singular matrix. It can be seen by inspection that the matrix $\bN$ in (\ref{LM}) is structured as $\Delta(\bN_1,\bN_2)$.

For the purposes of Section~\ref{sec:LQHO_PR}, the notion of specteral genericity is provided in the following definition.
\begin{definition} \cite{AKh_2015}
\label{def:gen}
The matrix $F$ and the state-space realization (\ref{equ:FGLK}) are said to be \emph{spectrally generic} if  the spectrum $\sigma(F)$ has no intersection with its mirror reflection
about the imaginary axis in the complex plane:
$
    \sigma(F)\bigcap \left(-\overline{\sigma(F)} \right) = \emptyset
$,
that is, $\lambda+\overline{\nu} \ne 0$ for all eigenvalues $\lambda,\nu\in \sigma(F)$.
\end{definition}
In view of the one-to-one correspondence described in this section, the matrix $F$, defined in (\ref{equ:FGLK}), is related to the matrix $A$, defined in (\ref{equ:ABCD}), by a similarity transformation. Hence, in the position-momentum form, spectral genericity is equivalent to the condition in which the spectrum $\sigma(A)$, which includes the poles of the associated transfer function, has no intersection with its mirror reflection about the origin of the complex plane.
\subsection{Cholesky-like Factorizations for Skew-Symmetric Matrices} 	\label{App:CLD}
\setcounter{equation}{0}
\label{App:Ch_dec}
 For the purposes of Section~\ref{sec:LQHO_PR}, the existence of Cholesky-like factorizations is addressed in the following lemma.
\begin{lem}
	\label{lem:ch_fact}
	Consider a non-singular matrix $\Theta \in \mA_{2n}$. There exists a non-singular matrix $\Sigma \in \mR^{2n \times 2n}$ such that $\Theta = \Sigma J_{2n} \Sigma^\rT$. 
\end{lem}
\begin{proof}
	As a consequence of the spectral decomposition, in the Murnaghan canonical form (see \cite{benner00} and references therein), there exists a factorization $\Theta = O \Delta O^\rT$, where the matrix $O \in \mR^{2n \times 2n}$ is orthogonal and the matrix $\Delta \in \mR^{2n \time 2n}$ is block diagonal. Each block on the main diagonal of the matrix $\Delta$ has the form $\scriptsize \begin{bmatrix} 0 & \delta_i\\ -\delta_i & 0 \end{bmatrix}$ with $\delta_i >0$, where $\pm i \delta_i$ is a pair of complex conjugate eigenvalues of $\Theta$. Then, there exists a decomposition $\Theta=\Sigma J_{2n} \Sigma^\rT$, where the matrix $\Sigma = O {\rm diag}\{ \sqrt{\delta_1}, \sqrt{\delta_1}, \hdots, \sqrt{\delta_n}, \sqrt{\delta_n} \} \Sigma_0$ is non-singular and $\Sigma_0$ is a permutation:
$
\Sigma_0 J_{2n} \Sigma_0^\rT
=
I_{n} \ox
{\scriptsize\begin{bmatrix}
     0 & 1\\
    -1 & 0
\end{bmatrix}}.
$
Also, for any such $\Sigma$, the matrix $\Sigma \wh{\Sigma}^\rT$ leads to the decomposition of $\Theta$, where $\wh{\Sigma} \in Sp(2n,\mR)$.
\end{proof}\noindent 
In view of Lemma~\ref{lem:ch_fact}, any two non-singular matrices $\Theta_1, \Theta_2 \in \mA_{2n}$  are related to each other by a non-singular matrix $\wh{\Sigma}$ as $\Theta_1=\wh{\Sigma} \Theta_2 \wh{\Sigma}^\rT$, where $\wh{\Sigma}= \Sigma_1 \Sigma_2^{-1}$ and $\Theta_k=\Sigma_k J_{2n} \Sigma_k^\rT$ for $k=1,2$.

\bibliographystyle{IEEEtran}
\bibliography{IEEEabrv,Biblist}
\end{document}